\documentclass[a4paper,11pt]{article}

\usepackage[english]{babel}
\usepackage{graphicx,charter}
\usepackage{framed}
\usepackage[normalem]{ulem}
\usepackage{amsmath}
\usepackage[ruled,vlined,linesnumbered]{algorithm2e}
\usepackage{amsthm,thm-restate,thmtools}
\usepackage{amssymb}
\usepackage{amsfonts}
\usepackage{enumerate}
\usepackage{etoolbox}
\usepackage{subcaption}
\usepackage[utf8]{inputenc}
\usepackage{fullpage}
\usepackage{xcolor,xfrac}
\usepackage{url}
\usepackage{enumitem}
\usepackage{tabularx}
\usepackage{algpseudocode}
\algdef{SE}[PROCEDURE]{Procedure}{EndProcedure}%
   [2]{\algorithmicprocedure\ \textproc{#1}\ifthenelse{\equal{#2}{}}{}{(#2)}}%
   {\algorithmicend\ \algorithmicprocedure}%

\usepackage{soul}
\usepackage{url}
\usepackage[hidelinks]{hyperref}

\usepackage{comment}

\newcommand{\A}{\mathcal{A}}
\newcommand{\G}{\mathcal{G}}

\DeclareMathOperator*{\argmax}{arg\,max}
\DeclareMathOperator*{\argmin}{arg\,min}

\renewcommand{\emptyset}{\varnothing}

\newtheorem{theorem}{Theorem}
\newtheorem{lemma}[theorem]{Lemma}
\newtheorem{corollary}[theorem]{Corollary}
\newtheorem{definition}{Definition}

\newcommand{\EFone}{$\mathrm{EF1}$}
\newcommand{\EFX}{$\mathrm{EFX}$}
\newcommand{\swapEF}{$\mathrm{swapEF}$}

\usepackage{xcolor}

\usepackage{natbib}
\date{}
\title{Repeatedly Matching Items to Agents Fairly and Efficiently}
\author{Ioannis Caragiannis\thanks{Department of Computer Science, Aarhus University, {\AA}bogade 34, 8200 Aarhus N, Denmark. Partially supported by the Independent Research Fund Denmark (DFF) under grant 2032-00185B. Email:~\url{iannis@cs.au.dk}}\and Shivika Narang\thanks{Department of Computer Science and Automation, Indian Institute of Science, Bengaluru, Karnataka, India. Supported by a Tata Consultancy Services research scholarship.  Email:~\url{shivika@iisc.ac.in}}}

\begin{document}
\maketitle

\begin{abstract}
\noindent We consider a novel setting where a set of items are matched to the same set of agents repeatedly over multiple rounds. Each agent gets exactly one item per round, which brings interesting challenges to finding efficient and/or fair {\em repeated matchings}. A particular feature of our model is that the value of an agent for an item in some round depends on the number of rounds in which the item has been used by the agent in the past. We present a set of positive and negative results about the efficiency and fairness of repeated matchings. For example, when items are goods, a variation of the well-studied fairness notion of envy-freeness up to one good (\EFone) can be satisfied under certain conditions. Furthermore, it is intractable to achieve fairness and (approximate) efficiency simultaneously, even though they are achievable separately. For mixed items, which can be goods for some agents and chores for others, we propose and study a new notion of fairness that we call {\em swap envy-freeness} (\swapEF). 
\end{abstract}

\section{Introduction}
The problem of fairly dividing indivisible items among agents has received enormous attention by the EconCS research community in the recent years. The standard setting involves a set of items and agents who have values for them. The objective is to compute an allocation which gives each item to a single agent so that some notion of fairness  is satisfied. A diverse set of fairness objectives has been explored in the past; some of the most well known of these are envy-freeness and its relaxations. Prior work has typically explored various settings where agents' allocations do not change with time. 

However, in some scenarios that arise in practice, the same set of items must be allocated to the same set of agents repeatedly.  More crucially, another feature that distinguishes such scenarios from the standard setting is that the value of an agent for an item changes over time and typically depends on how many times the agent has received the item in the past. This can make solutions that were fair when the agents were allocated the items once, no longer fair. 

To give an example, consider different research labs that all need access to several expensive research facilities in a university. How should the access of the labs to the facilities be fairly coordinated/scheduled throughout the year? This is a fair division problem with the labs and the facilities playing the role of the agents and the items, respectively. To be fair among labs and efficient overall, such a scheduling should take into account the values the  labs have for facilities, which typically change over time. For instance, during the first few weeks of access to a facility, the researchers in a lab may need time to learn how to operate it. During that time, the value the lab gets by accessing a facility can be very low, even negative. As the researchers gain more experience, their research output increases, and so does the lab's value for the facility. Once the researchers have run their intended experiments, the lab's value for the facility decreases again until the next experiment. 
We model such situations with $n$ agents who must be matched with exactly one of $n$ items in each of $T$ rounds, repeatedly. The value an agent has for an item in a round depends on how many times the agent has used the item in previous rounds. We use {\em social welfare} (the total value of the agents from the items they get in all rounds) to assess the efficiency of such {\em repeated matchings}. We also use  relaxations of {\em envy-freeness} as fairness concepts. In particular, we adapt the well-known {\em envy-freeness up to one item} (EF1) and use it when all valuations are non-negative (i.e., when items are {\em goods}). A repeated matching is \EFone\ if the value of every agent $i$ for her bundle is at least as high as her value for the bundle of any other agent $j$ after removing the last {\em copy} of an item from $j$'s bundle. We observe that \EFone\ is not suitable when valuations can be positive or negative (i.e., when items are {\em mixed}), and introduce the new notion of {\em swap envy-freeness} to assess fairness of repeated matchings for mixed items. 

\subsection{Our Contribution} 
More specifically, our technical contribution is as follows. We prove that the problem of computing a repeated matching with maximum social welfare is NP-hard, even when $T=3$. Our hardness reduction defines instances with items and non-monotone valuations. 
The problem becomes solvable in polynomial time when the valuations are monotone. This is when the value an agent has for an item can only decrease or increase, but not both, in terms of the number of rounds the agent had the item in the past.  For the case of monotone non-increasing valuations, earlier work on $b$-matchings can be leveraged to find the optimal solution. When the valuations are monotone non-decreasing, we find a neat reduction  to the case of time-constant valuations which can be solved efficiently. 

We also consider fair repeated matchings, using \EFone\ as fairness concept. We find that under identical valuations, \EFone\ repeated matchings always exist and can be found in polynomial time. Furthermore, we show that any instance with general valuations and $T\bmod n\in\{0,1,2,n-1\}$ (i.e., including all instances with at most four agents/items) has an \EFone\ repeated matching, which can be computed efficiently. We establish that, unfortunately, \EFone\ is not compatible with social welfare maximization and even approximating the maximum social welfare over \EFone\ repeated matchings is NP-hard. This holds even for settings where \EFone\ solutions can be found in polynomial time. 

Moreover, at a conceptual level, we propose and study a new fairness notion called swap envy-freeness (\swapEF). Here, we find that under identical valuations, \swapEF\ repeated matchings can be found using the same algorithm as used for \EFone. Furthermore, we show that \swapEF\ repeated matchings always exist and can be computed efficiently on instances with $T\bmod n\in\{0,1,2,n-2,n-1\}$ (i.e., including all instances with at most five agents/items). Our hardness results are proved on instances with goods. Our positive results besides those for \EFone, apply to instances with mixed items. 

\subsection{Related Work} In fair division with indivisible items, \EFone\ has been established as a key fairness concept. It was defined by~\cite{budish2011combinatorial} (and, implicitly, a few years earlier by~\citealp{lipton2004approximately}). In contrast to envy-freeness which is usually impossible to achieve, \EFone\ is always achievable in the standard setting and is also compatible with notions of economic efficiency~\citep{caragiannis2019unreasonable,BKV18}. These papers assume that items are goods, i.e., agents have non-negative valuations for them. Non-positive valuations, i.e., indivisible chores, have also received attention. More importantly, a series of recent papers consider mixed items that can be goods for some agents and chores for others~\citep{ACIW22,BSV21,BBB+20}.

The main assumption in the standard setting is that each item is given to exactly one agent with no explicit cap on the number of items one agent can get. 
\citet{biswas2018fair} consider an extension where the items are partitioned into categories and there are cardinality constraints on how many items an agent can be allocated from each category. They show how to compute an \EFone\ allocation by extending the envy-cycle elimination algorithm of~\citet{lipton2004approximately}. Even though cardinality constraints can restrict allocations to repeated matchings, our history-dependent valuations cannot be expressed by their model. Another extension is considered by \cite{gafni2021unified} where each item may have multiple copies. They study relaxations of envy-freeness with mixed items, in a model where each agent can get at most one item copy.

The concept of repeated matching has been considered before, actually using \EFone\ as fairness concept. \citet{HLC15} look at a dynamic one-sided repeated matching model with ordinal preferences that change over time. They study strategyproofness and give a mechanism that is \EFone. As the model of preferences studied is entirely different, their results are not applicable to our model. \citet{gollapudi2020almost} study a two-sided repeated matching setting where the agent values may change in each round, but do not take into account how often the two agents have been matched in the past. In addition, due to the two-sided nature of their setting, their results are not applicable to our case. 

Finally, relaxations of envy-freeness have been considered extensively in the literature. For mixed items in particular, \cite{A20} summarizes the several variations of \EFone\ that have been proposed in the literature and proposes new ones. The setting of identical valuations has also been specifically explored, both for the existence of approximately envy-free solutions \citep{plaut2020almost,chen2020fairness} and other objectives \citep{barman2018greedy,BarmanS20}. To the best of our knowledge, swap envy-freeness appears to be novel.

\subsection{Roadmap} The rest of the paper is structured as follows. We begin with setting up the notation and relevant definitions in Section \ref{sec:prelim}. Section \ref{sec:max-sw} focuses on maximizing social welfare. Here, we give our hardness result for maximizing social welfare in general and polynomial-time algorithms for monotone valuations. In Section \ref{sec:fairness}, we explore settings under which we can satisfy \EFone\ and algorithms that find \EFone\ solutions. In Section \ref{sec:swef1}, we find that even in settings where \EFone\ repeated matchings can be found in polynomial time, maximizing social welfare over the space of \EFone\ repeated matchings is intractable. We devote Section~\ref{sec:swapEF} to the study of swap envy-freeness. We conclude with a discussion on open problems in Section~\ref{sec:open}.

\section{Notation and Preliminaries}\label{sec:prelim}
Our setting involves a set $\A$ of $n$ agents and a set $\G$ of $n$ items. We use the term {\em matching} to refer to an allocation of the items to the agents, so that each agent gets exactly one item and each item is given to exactly one agent. We particularly focus on {\em repeated matchings}, where the items are matched to the agents in multiple rounds. More formally, we consider instances of the form $I=\langle \A,\G,\{v_i\}_{i\in \A},T\rangle$, where $T$ denotes the number of rounds and, for each agent $i\in \A$, $v_i$ is a function from $\G\times [T]$ to $\mathbb{R}$, where $v_i(g,t)$ denotes the {\em valuation} of agent $i$ for item $g$ when it is matched to the item for the $t^{\text{th}}$ time. A repeated matching $A=(A^1, ..., A^T)$ is simply a collection of matchings, with one matching $A^t$ per each round $t\in [T]$. Furthermore, we denote by $A_i$ the multiset (or {\em bundle}) which contains copies of the items to which agent $i\in\A$ is matched in the $T$ rounds.

Hence, defining the bundles $A_i$ for $i\in \A$ given the repeated matching $A$ is trivial. The opposite task is also straightforward. Let $N(B,g)$ be the multiplicity of item $g$ in bundle $B$. Given bundles of items $A_i$ for $i\in \A$ with $|A_i|=T$ (i.e., each agent gets $T$ copies of items) and $\sum_{i\in \A}{N(A_i,g)}=T$ (i.e., $T$ copies of each item $g$ are allocated), a consistent repeated matching\footnote{We remark that this repeated matching is not unique. However, this does not affect the values of each agent for her bundle and the bundle of any other agent, which are the same in all different consistent repeated matchings.} for instance $I$ is obtained as follows. We construct the bipartite multigraph $G=(\A,\G,E)$ so that the set of edges $E$ consists of (a copy of) edge $(i,g)$ for every (copy of) item $g$ such that $g\in A_i$. The graph $G$ is $T$-regular and, thus, by Hall's matching theorem (see~\citealp{PL86}), can be decomposed into $T$ matchings of edges $M_1$, ..., $M_T$. These matchings correspond to a repeated matching by interpreting the edge $(i,g)$ in matching $M_t$ as the assignment of item $g$ to agent $i$ in the $t^\text{th}$ round. 

With a slight abuse of notation, we use $v_i(B)$ to denote the {\em value} agent $i\in \A$ has when she gets the bundle $B$, i.e., 
\[v_i(B)=\sum_{g\in \G}\sum_{t=1}^{N(B,g)}v_i(g,t).\]
Hence, for a repeated matching $A$, $v_i(A_i)$ is the total value from each item copy agent $i$ receives in all rounds. The {\em social welfare} of $A$ is simply the sum of the agents' values for their bundle, i.e., $SW(A)=\sum_{i\in \A}v_i(A_i)$.

We shall look at specific types of valuations under which we will try to find efficient and/or fair repeated matchings. A well-motivated setting is that of {\em identical} valuations where $v_1=v_2=\cdots =v_n$. This assumption proves particularly useful in finding fair solutions. Another important class of valuation functions is that of {\em monotone} valuations.

\begin{definition}[monotone valuations]
The valuation function $v_i$ is monotone non-increasing (respectively, monotone non-decreasing) if for every item $g\in \G$, and $t\in [T-1]$, we have that $v_i(g,t)\geq v_i(g,t+1)$ (respectively, $v_i(g,t)\leq v_i(g,t+1)$).
\end{definition}

\noindent These two classes of valuation functions intersect in the class of {\em constant} valuations.

\begin{definition}[constant valuations]
Valuation function $v_i$ is said to be constant if for every item $g\in \G$,  we have that $v_i(g,1)=v_i(g,2)=\cdots =v_i(g,T)=v_i(g)$.
\end{definition}

We extend to repeated matchings the well-known fairness notion of {\em envy-freeness of up to one item} (EF1) as follows. 

\begin{definition}[\EFone]
A repeated matching $A$ is \EFone\ if for every pair of agents $i,j \in \A$, there exists an item $g \in \G$ such that $v_i(A_i)\geq v_i(A_j\setminus \{g\})$.
\end{definition}
\noindent We remark that the operation $A_j\setminus \{g\}$ removes one copy of item $g$ from the bundle $A_j$ if $g$ belongs to $A_j$ and leaves $A_j$ intact otherwise.

We refer to the items as {\em goods} on instances where all valuations are non-negative, i.e., when $v_i(g,t)\geq 0$ for every $i\in \A$, $g\in\G$, and $t\in [T]$. When there are no restrictions on the valuations, we refer to the items as {\em mixed}. 

\section{Maximizing Social Welfare}\label{sec:max-sw}
We begin our technical exposition by studying the complexity of the problem of computing a repeated matching of maximum social welfare. Notice that if $T=1$, this task can be easily done by computing a maximum-weight perfect matching in the complete bipartite graph $G=(\A,\G,\A\times\G)$, in which edge $(i,g)$ has weight $v_i(g,1)$. For $T>1$, an approach that seems natural computes gradually a maximum-weight perfect matching for each round, taking into account the matching decisions in previous rounds. 

For example, consider the instance with three agents and two rounds (i.e., $n=3$, $\A=\{1,2,3\}$, $\G=\{g_1,g_2,g_3\}$, and $T=2$). The agent valuations are as follows: $v_1(g_2,1)=v_1(g_3,1)=1-\epsilon$ (for small but strictly positive $\epsilon$), $v_2(g_2,1)=v_3(g_3,1)=1$, while all other valuations are $0$. A maximum-weight perfect matching on the complete bipartite graph $G=(\A,\G,\A\times\G)$ with weight $v_i(g,1)$ on edge $(i,g)$ assigns item $g_i$ to agent $i$ in the first round; this gives value $1$ to agents $2$ and $3$. Then, the natural way to compute the matching of the second round is to compute a maximum-weight perfect matching in the complete bipartite graph $G=(\A,\G,\A\times\G)$ with weight $v_i(g_i,2)$ to edge $(i,g_i)$ (because agent $i$ already uses item $i$ in the first round) and weight $v_i(g,1)$ to edge $(i,g)$ for $g\not=g_i$. In this way, the matching of the second round will give value of $1-\epsilon$ to agent $1$ only, by matching her to either item $g_2$ or item $g_3$. Thus, the social welfare is $3-\epsilon$. In contrast, consider the repeated matching in which the first-round matching assigns item $g_2$ to agent $1$, item $g_1$ to agent $2$, and item $g_3$ to agent $3$, and the second-round matching assigns item $g_3$ to agent $1$, item $g_2$ to agent $2$, and item $g_1$ to agent $3$. Agent $1$ gets value $1-\epsilon$ in both rounds, agent $2$ gets value $1$ in the second round, and agent $3$ gets value $1$ in the first round. Hence, the social welfare is now $4-2\epsilon$, i.e., higher than before.

This example demonstrates that computing a repeated matching of maximum social welfare can be a challenging task. Actually, as our first result indicates, the problem is hard.

\begin{theorem}\label{thm:sw-hard}
Given a repeated matching instance, computing a repeated matching of maximum social welfare is NP-hard.
\end{theorem}

\begin{proof}
We present a polynomial-time reduction from exact $3$-cover (X3C). An instance of X3C consists of a universe $U=[3q]$ of elements and a collection $\mathcal{S}$ of $m$ sets $S_1, S_2, ..., S_m$, containing three elements of $U$ each. Deciding whether there are $q$ mutually disjoint sets in $\mathcal{S}$ is a well-known NP-hard problem; e.g., see~\citet{GJ79}.

We construct an instance $I=\langle \A,\G, \{v_i\}_{i\in \A},T\rangle$ with $T=3$, and $n=m+3q$ agents/items. The set of agents $\A$ has a {\em set agent} $i$ for every $i\in [m]$ and an {\em element agent} $m+i$ for every element $i$ of $U$. The set of items $\G$ has an {\em element item} for every element $i\in [U]$, $m-q$ {\em space-filling} items $3q+1, ..., m+2q$, and $q$ {\em dummy items} $m+2q+1, ..., m+3q$. The valuations are as follows:
\begin{itemize}
    \item For every $i\in [m]$ and every element $g\in S_i$, the set agent $i$ has value $v_i(g,1)=1$ for the first copy of element item $g$.
    \item For every $i\in [m]$ and every $g=3q+1, ..., m+2q$, the set agent $i$ has value $v_i(g,3)=3$ for the third copy of the space-filling item $g$.
    \item For every $i\in [3q]$, the element agent $m+i$ has value $v_{m+i}(i,2)=3$ for the second copy of the element item $i$.
    \item All other valuations are $0$ (including the valuation of any  agent for a dummy item).
\end{itemize}

We claim that there are $q$ mutually disjoint sets in $\mathcal{S}$ if and only if there is a repeated matching of social welfare $3m+9q$ in $I$. We begin by presenting a repeated matching of social welfare $3m+9q$ when $\mathcal{S}$ has a subcollection $X$ of $q$ disjoint sets. For every $i$ such that $S_i\in X$, the set agent $i$ gets one copy of each element item corresponding to an element $g\in S_i$. Agent $i$ gets a value of $3$ in this way. For every $i$ such that $S_i\not\in X$ (notice that there are exactly $m-q$ such $i$'s), the set agent $i$ gets three copies of a distinct space-filling item. Hence, the set agents who do not get element items have value $3$, too. For $i\in [3q]$, the element agent $m+i$ gets two copies of the element item $i$. Again, the element agents have all value $3$. Each of the $3q$ element agents gets one distinct copy of one of the $q$ dummy items; these do not contribute to the social welfare.

Now consider a repeated matching on instance $I$ that has social welfare $3m+9q$. This means that each of the $m-q$ space-filling items gives value $3$ to the agents while each of the $3q$ element items gives them value $4$. Notice that these are the maximum contributions from each item to the social welfare. The only way that each space-filling item gives a value of $3$ to the agents is when all its three copies are given to the same set agent. Hence, $m-q$ of the set agents have three copies of a space-filling item each. Also, the only way for an element item $g$ to give value $4$ to the agents is when two of its copies are given to the element agent $m+g$ and another copy is given to a set agent $i$ such that $g\in S_i$. Hence, every set agent $i$ who does not include a space-filling item contains a single copy of each of the three items corresponding to the elements in $S_i$ which is not used in any other set agent. Hence, the union of the $q$ sets corresponding to these set agents includes all elements of $U$.
\end{proof}

\subsection{Monotone Valuations}\label{subsec:welfare-non-increasing}
Fortunately, the problem can be solved in polynomial time for monotone valuations, even when the items are mixed. Notice that the instance in the example given at the beginning of Section~\ref{sec:max-sw} belongs to the category of monotone non-increasing valuations. 

\paragraph{Monotone non-increasing valuations.}
For this particular case, well-known results on $b$-matchings can be used to find a social welfare maximizing repeated matching. In the following, we briefly explain how; recall that a $b$-matching in a bipartite graph is just a subset of the edges that includes at most $b$ edges that are incident to any given node. \citet{GT89} show how to compute a maximum-weight $b$-matching on input an edge-weighted bipartite multigraph in time that is polynomial in $b$, the size of the graph, and the number of bits required to represent the edge-weights.

Given a repeated matching instance $I=\langle \A, \G, \{v_i\}_{i\in \A},T \rangle$ where each $\{v_i\}_{i\in \A}$ is monotone non-increasing, construct the bipartite multigraph graph $G=(\A,\G, E)$ where $E$ consists of $T$ copies of edge $(i,g)$ for each $i\in\A$ and each $g\in\G$. For each $t\in [T]$, $i\in \A$ and $g\in \G$, we set the edge weight of the $t^{\text{th}}$ copy of edge $(i,g)$ to $v_i(g,t)$. Now, since the $v_i$s are monotone non-increasing, we can assume that a maximum-weight $T$-matching in $G$ has the following {\em consecutive edge copies} property: if it contains $k$ copies of an edge $(i,g)$, these are the first $k$ copies of weights $v_i(g,1)$, ..., $v_i(g,k)$. Notice that, if this is not the case, we can redistribute the edge copies of $(i,g)$ between agents appropriately without violating weight maximality. 

Now, a maximum-weight $T$-matching $M$ in $G$ naturally defines a repeated matching $A_M$ in $I$, where each $i$ is matched to each $g$ as many times as the number of copies of edge $(i,g)$ $M$ contains. Furthermore, the social welfare of $A_M$ is equal to the weight of $M$ and can be seen to be optimal. The reason is that any repeated matching corresponds to a $T$-matching with the consecutive edge copies property. 

In Appendix~\ref{sec:swnoninc}, we present an alternative approach for finding social welfare maximizing repeated matchings for instances with monotone non-increasing valuations. The main idea is to formulate the problem as an integer linear program and use an LP solver to compute an extreme solution of the LP relaxation, which, as we show, is guaranteed to be integral.

\paragraph{Monotone non-decreasing valuations.} Neither 
$b$-matchings nor our linear programming-based approach can be used when all the valuation functions are monotone non-decreasing. Somewhat surprisingly, it suffices to resort to an even simpler ordinary matching computation in this case. 

We remark that, on repeated matching instances with constant valuations, there is always a repeated matching of maximum social welfare in which every agent gets the same item in all rounds. To see why, consider any repeated matching $A$ and let $t$ be that round in which the total value the agents get from the items they get in matching $A^t$ is maximum. Then, the repeated matching which uses matching $A^t$ in all rounds has at least as high social welfare with $A$. Hence, a straightforward maximum-weight matching computation can be used to compute a social welfare maximizing repeated matching for instances with constant valuations. The proof of the next theorem  exploits a connection of instances with monotone non-decreasing valuations and instances with constant valuations.

\begin{theorem}\label{thm:sw-non-dec}
Given a repeated instance with monotone non-decreasing valuations, a repeated matching of maximum social welfare can be computed in polynomial time.
\end{theorem}

\begin{proof}
Consider a repeated matching instance $I=\langle \A, \G, \{v_i\}_{i\in \A}, T\rangle$ with monotone non-decreasing valuations. For each agent $i\in \A$, we construct the constant valuation function $v_i^c$ with $v_i^c(g)=\frac{1}{T}\sum_{t=1}^T{v_i(g,t)}$ for each item $g\in \G$. That is, the value that agent $i$ gets from a copy of item $g$ under valuation $v_i^c$ is the average value that $i$ gets from $g$ under $v_i$ in $T$ rounds.

Observe that, by the definition of the valuation $v_i$, it holds $v_i(A_i)\leq v_i^c(A_i)$ for any repeated matching $A$ and any agent $i\in \A$. This implies that the social welfare of $A$ under the valuations $v_i$ is not higher than the social welfare under the valuations $v_i^c$. Hence, the maximum social welfare among all repeated matchings with respect to valuations $v_i$ is not higher than the maximum social welfare among all repeated matchings with respect to valuations $v_i^c$. Furthermore, the maximum social welfare under $v_i^c$ is achieved by a repeated matching $\widehat{A}$ that uses the same matching in all rounds. The proof completes by observing that $v_i(\widehat{A}_i)= \sum_{t=1}^T{v_i(g_i,t)}=v_i^c(\widehat{A}_i)$, where $g_i$ is the item agent $i$ gets in all rounds under $\widehat{A}$. I.e., the social welfare of $\widehat{A}$ is the same with respect to the original valuations $v_i$ and the modified valuations $v_i^c$. Thus, to maximize the social welfare, it suffices to compute a single-round matching of maximum social welfare according the valuations $v_i^c$ and repeat it for $T$ rounds.
\end{proof}

\section{Computing Fair Repeated Matchings}\label{sec:fairness}
In this section, we focus on repeated matching instances with goods (i.e., non-negative valuations) and present algorithms that compute \EFone\ repeated matchings under different conditions. We begin by considering identical valuations in Section~\ref{subsec:identical} and conclude  with our results for general non-negative valuations in Section~\ref{subsec:general}. As an interlude, we discuss how previous work can be adapted to repeated matching instance with constant valuations (Section~\ref{subsec:constant}).

\subsection{Identical Valuations}\label{subsec:identical}
Our algorithm for repeated matching instances with identical valuations works as follows. It starts by assigning $\lfloor T/n\rfloor$ copies of each item to each agent. If $T\bmod n>0$ (i.e., additional copies have to be assigned to the agents so that the repeated matching is correct), the algorithm works in a round robin fashion for $T\bmod n$ phases. In these phases, it uses a fixed ranking of the items according to the value $v(g\lceil T/n\rceil)$ of their $\lceil T/n\rceil$-th copy. The ranking assigns to each item a distinct integer $\text{rank}(g)$ in $[n]$ such that $\text{rank}(g_1)<\text{rank}(g_2)$ implies that $v(g_1,\lceil T/n\rceil)\geq v(g_2,\lceil T/n\rceil)$. In each round-robin phase, the agents act according to the ordering $1, 2, ..., n$. When it is agent $i$'s turn, she picks a copy of the lowest-rank item that is available.

The algorithm appears below as Algorithm~\ref{alg:identical}. It has access to function $\text{rank}()$ defined as above and uses the matrix $f$ to store the number of copies of each item an agent gets. The final step is to call routine $\text{GenerateFromFreq}()$ to transform $f$ to the repeated matching $A$; this routine essentially implements the transformation described in Section~\ref{sec:prelim} and is called at the final step of every algorithm we present in the paper.

\begin{algorithm}[!ht]
    \KwIn{Identical Valuations Instance $I=\langle \A,\G,v,T\rangle$ with $|\A|=n$}
    \KwOut{A repeated matching $A$}
    $f(i,g)\gets \lfloor \sfrac{T}{n} \rfloor , \, \forall i\in \A,\,\forall g\in \G$\;\label{algid:L2}
    \If{$T \bmod n>0$}{
        $x_g\gets T \bmod n, \, \forall g\in G$\;
        \For{$t=1$ to $T \bmod n$}{
            \For{$i=1$ to $n$}{
                $g'\gets \argmin_{g:\, x_g >0}\text{rank}(g)$\;
                $x_{g'} \gets x_{g'} -1 $\;
                $f(i, g') \gets \lceil \sfrac{T}{n} \rceil $\;
            }
        }
    }
    $A\gets \text{GenerateFromFreq}(f)$\;

\caption{Computing an \EFone\ repeated matching under identical valuations} \label{alg:identical}
\end{algorithm}

We now use Algorithm~\ref{alg:identical} to prove the next statement.

\begin{theorem}
Given a repeated matching instance with identical valuations, an EF1 repeated matching exists and can be computed in polynomial time.
\end{theorem}

\begin{proof}
Algorithm~\ref{alg:identical} clearly runs in polynomial time. It remains to prove that it always returns an \EFone\ repeated matching.  Consider its application to a repeated matching instance $I=\langle \A,\G,v,T\rangle$, where $v$ is non-negative. The repeated matching returned is clearly \EFone\ if $T$ is an integer multiple of $n$; in this case, all agents get the same number of copies of all items and nobody is envious. 

Otherwise, since $T \bmod n \leq n-1$ copies of each item are available in the round-robin phases and all the remaining $T \bmod n$ copies of each item are picked in consecutive round-robin steps, no agent gets more than one copy of the same item in the round robin phases. Let $g_{i,t}$ be the item agent $i$ gets in the round robin phase $t\in \{1, 2, ..., T\bmod n\}$. Consider two agents $i$ and $j$ and observe that the repeated matching $A$ returned by Algorithm~\ref{alg:identical} satisfies
\begin{align*}
    v(A_i)-v(A_j\setminus\{g_{j,1}\}) &=\sum_{t=1}^{(T\bmod n)-1}{(v(g_{i,t},\lceil T/n\rceil)-v(g_{j,t+1},\lceil T/n\rceil))}+v(g_{i,T\bmod n})\\
    &\geq \sum_{t=1}^{(T\bmod n)-1}{(v(g_{i,t},\lceil T/n\rceil)-v(g_{j,t+1},\lceil T/n\rceil))} \geq 0,
\end{align*}
as \EFone\ requires. The equality follows since both agents $i$ and $j$ get $\lfloor T/n \rfloor$ copies of each item at the beginning of the algorithm and, then, the valuation difference is due to the $\lceil T/n\rceil$-th copies of items allocated in the round-robin phases. The first inequality is due to the non-negativity of valuations. The second one follows since the item that agent $i$ picks at the round-robin phase $t$ has not higher rank than the item agent $j$ picks in the next phase $t+1$.
\end{proof}

\subsection{Constant Valuations}\label{subsec:constant}
Before discussing the case of general valuations, we briefly discuss the seemingly related problem of fair division with cardinality constraints and whether existing results can be used to obtain \EFone\ repeated matchings in our case. \citet{biswas2018fair} consider an extension of the standard fair division setting where a set of items (goods) need to be allocated to a set of agents with additive valuations for the items. The additional feature of their problem is that the set of items is partitioned into categories and each category has a cardinality constraint. The objective is now to compute an allocation of the items to the agents, in which the number of items each agent gets from each category does not exceed the cardinality constraint of that category. \citet{biswas2018fair} show that allocations that satisfy such cardinality constraints and are furthermore \EFone\ do exist and can be computed in polynomial time.

Notice that the results of \citet{biswas2018fair} can be used to compute \EFone\ repeated matchings for instances with constant valuations. Indeed, given a repeated matching instance $I=\langle \A,\G,\{v_i\}_{i\in \A},T)$, it suffices to consider a fair division instance $I'$ with the $n$ agents in $\A$, $T$ distinct items for each item $g$ in $\G$, each of value $v_i(g)$ to agent $i$, and a cardinality constraint of $T$ for the whole set of items. It can be easily seen that any \EFone\ allocation for instance $I'$ naturally corresponds to an \EFone\ repeated matching for instance $I$ and vice versa. Unfortunately, for non-constant valuations, this reduction does not work as it seems impossible to express the history-dependent valuations in our model with additive valuations for items in the model of \citet{biswas2018fair}. 

An algorithmic idea for repeated matchings that is inspired by Algorithm~\ref{alg:identical} is to begin by assigning $\lfloor T/n \rfloor$ copies of each item to each agent and distribute the remaining $T\bmod n$ copies of each item so that each agent gets at most one additional copy. Can we achieve \EFone\ in this way for general valuations? This requires the computation of an \EFone\ repeated matchings on instances with $T<n$, in which each agent gets at most one copy of each item (and $T$ copies in total). Even though additivity would not be a problem anymore, it is still not clear how to express such instances in the model of \citet{biswas2018fair} using cardinality constraints defined on a {\em single} partition of the items only.

\subsection{General Valuations}\label{subsec:general}
We now prove that \EFone\ repeated matchings can be computed in polynomial time for general non-negative valuations when the number $T$ of rounds and the number $n$ of agents/items satisfy a particular condition. 

\begin{theorem}\label{thm:EF1}
Given a repeated matching instance $I$ with $n$ agents/goods and $T$ rounds such that $T\bmod n \in\{0,1,2,n-1\}$, an EF1 repeated matching exists and can be computed in polynomial time.
\end{theorem}

We prove Theorem~\ref{thm:EF1} constructively, by defining two algorithms for the cases $T\bmod n\in \{0,1,2\}$ (Algorithm~\ref{alg:ef1-1}) and $T\bmod n =n-1$ (Algorithm~\ref{alg:ef1-2}). ~

Algorithm~\ref{alg:ef1-1} computes the number of copies of each item  that each agent gets as follows. First, it gives to each agent $\lfloor T/n\rfloor$ copies of each item (line~\ref{alg1:L2}). If $T\bmod n \not=0$, it then runs a round-robin phase (lines~\ref{alg1:L3}-\ref{alg1:L8}) and then, if $T\bmod n=2$, it runs an additional reverse round-robin phase (lines~\ref{alg1:L9}-\ref{alg1:L14}). In the round-robin phase, the agents act according to the ordering $1, 2, ..., n$ (see the for-loop in lines~\ref{alg1:L5}-\ref{alg1:L8}). When it is agent $i$'s turn to act, she gets the item $\widehat{g}$ (identified in line~\ref{alg1:L6}) for which her value for the $\lceil T/n\rceil$-th copy is maximum among the items that have not been given to agents who acted before $i$ in the round-robin phase (the set variable $P$ is used to identify these items). In the reverse round-robin phase, the agents act according to the ordering $n, n-1, ..., 1$ (see the for-loop in lines~\ref{alg1:L11}-\ref{alg1:L14}). When it is agent $i$'s turn to act, she gets the item $\widehat{g}$ (identified in line~\ref{alg1:L12}) for which her value for the next copy is maximum among the items that have not been given to agents who acted before $i$ in the reverse round-robin phase. Finally, the algorithm transforms the matrix $f$ indicating the number of copies of each item each agent gets to a repeated matching by calling routine $\text{GenerateFromFreq}()$. Algorithm~\ref{alg:ef1-1} clearly runs in polynomial time. Its correctness is given by the next lemma.

\begin{algorithm}[!ht]
    \KwIn{Instance $I=\langle \A,\G,\{v_i\}_{i\in \A},T\rangle$ with $|\A|=n$ and $T\bmod n \in \{0,1,2\}$} 
    \KwOut{A repeated matching $A$}
    $f(i,g)\gets \lfloor T/n\rfloor, \forall i\in \A,\forall g\in G$\; \label{alg1:L2}
    \If{$T\bmod n>0$}{\label{alg1:L3}
        $P\gets \G$\;
        \For{$i=1$ to $n$}{\label{alg1:L5}
            $\widehat{g}\gets \argmax_{g\in P} v_i(g, \lceil T/n\rceil)$\;\label{alg1:L6}
            $f(i,\widehat{g})\gets \lceil T/n\rceil$\;
            $P\gets P \setminus \{\widehat{g}\}$\; \label{alg1:L8}
        }
    }
    \If{$T\bmod n=2$}{\label{alg1:L9}
        $P\gets \G$\;
        \For{$i=n$ to $1$}{\label{alg1:L11}
            $\widehat{g}\gets \argmax_{g\in P} v_i(g, f(i,g)+1)$\;\label{alg1:L12}
            $f(i,\widehat{g})\gets f(i,\widehat{g})+1$\;
            $P\gets P \setminus \{\widehat{g}\}$\;\label{alg1:L14}
        }
    }     
    $A\gets \text{GenerateFromFreq}(f)$\;

\caption{Computing an \EFone\ repeated matching} \label{alg:ef1-1}
\end{algorithm}

\begin{lemma}\label{lem:alg1}
The repeated matching $A=(A_1, ..., A_n)$ produced by Algorithm~\ref{alg:ef1-1} is \EFone.
\end{lemma}

\begin{proof}
Let $S$ denote the multiset that contains each item with multiplicity $\lfloor T/n\rfloor$. If $T\bmod n=0$, then $A_i=S$ for every agent $i$ and, hence, agents are not envious of each other. If $T\bmod n=1$, the final repeated matching is obtained after the execution of the round-robin phase. Consider two agents $i$ and $j$. Denoting by $g_j$ the item agent $j$ gets in this phase, agent $i$ has value $v_i(A_i)\geq v_i(S) = v_i(A_j\setminus \{g_j\})$, i.e., she satisfies the \EFone\ condition. 

If $T\bmod n=2$, the final repeated matching is obtained after the execution of the reverse round-robin phase. Consider two agents $i$ and $j$ with $i<j$. Let $g_i^1$ and $g_j^1$ be the items the agents $i$ and $j$ get in the round-robin phase and $g_i^2$ and $g_j^2$ be the items they get in the reverse round-robin phase, respectively. Agent $i$ has value 
\begin{align*}
v_i(A_i) &\geq v_i(S)+v_i(g_i^1,\lceil T/n \rceil) \geq v_i(S)+v_i(g_j^1,\lceil T/n \rceil) = v_i(A_j\setminus \{g_j^2\}).
\end{align*}
The second inequality follows since agent $i$ prefers item $g_i^1$ to item $g_j^1$ in the round-robin phase. For agent $j$, we distinguish between two cases. Let $\mu$ denote the multiplicity of item $g_j^2$ in $A_j$. If $g_j^1\not=g_i^2$, we have that, in the reverse round-robin phase, agent $j$ prefers the $\mu$-th copy of $g_j^2$ to the $\lceil T/n \rceil$-th copy of $g_i^2$, i.e., $v_j(g_j^2,\mu)\geq v_j(g_i^2,\lceil T/n\rceil)$. Then, we have 
\begin{align*}
v_j(A_j) &\geq v_j(S)+v_j(g_j^2,\mu) \geq v_j(S)+v_j(g_i^2,\lceil T/n\rceil)=v_j(A_i\setminus \{g_i^1\}). 
\end{align*}
If $g_j^1=g_i^2$, we have 
\begin{align*}
v_j(A_j) &\geq v_j(S)+v_j(g_j^1,\lceil T/n\rceil) = v_j(S)+v_j(g_i^2,\lceil T/n\rceil)=v_j(A_i\setminus \{g_i^1\}).
\end{align*}
Thus, the \EFone\ conditions for agents $i$ and $j$ are satisfied.
\end{proof}

Algorithm~\ref{alg:ef1-2} uses a similar structure. It starts by giving $\lceil T/n\rceil$ copies of each item to each agent (in line~\ref{alg2:L2}) and then removes the copy of a distinct item from each agent by running a round-robin phase (lines~\ref{alg2:L3}-\ref{alg2:L7}). 
When it is agent $i$'s turn to act, she gets rid of a copy of the item $\widehat{g}$ (identified in line~\ref{alg2:L5}) for which her value for the $\lceil T/n\rceil$-th copy is minimum among the items that have not been gotten rid by agents who acted before $i$ in the round-robin phase.

\begin{algorithm}[!ht]
    \KwIn{Instance $I=\langle \A,\G,\{v_i\}_{i\in \A},T\rangle$ with $|\A|=n$ and $T\bmod n= n -1$}
    \KwOut{A repeated matching $A$}
    $f(i,g)\gets \lceil T/n\rceil, \forall i\in \A,\forall g\in G$\;\label{alg2:L2}
    $P\gets \G$\;\label{alg2:L3}
    \For{$i=1$ to $n$}{\label{alg2:L4}
        $\widehat{g}\gets \argmin_{g\in P} v_i(g, \lceil T/n\rceil)$\;\label{alg2:L5}
        $f(i,\widehat{g})\gets f(i,\widehat{g})-1$\;\label{alg2:L6}
        $P\gets P \setminus \{\widehat{g}\}$\;\label{alg2:L7}
    }
    $A\gets \text{GenerateFromFreq}(f)$\;

\caption{Computing an \EFone\ repeated matching} \label{alg:ef1-2}
\end{algorithm}

\begin{lemma}\label{lem:alg2}
The repeated matching $A=(A_1, ..., A_n)$ produced by Algorithm~\ref{alg:ef1-2} is \EFone.
\end{lemma}

\begin{proof}
Let $i$ and $j$ be two agents and denote by $g_i$ and $g_j$ the items that are removed from their bundles in the round-robin phase. We have \begin{align*}
v_i(A_i) &=v_i(A_j)+v_j(g_j,\lceil T/n \rceil)-v_i(g_i,\lceil T/n \rceil) \geq v_i(A_j)-v_i(g_i,\lceil T/n \rceil)=v_i(A_j\setminus\{g_i\}),
\end{align*}
as desired. The last equality follows since $A_j$ has exactly $\lceil T/n \rceil$ copies of item $g_i$.
\end{proof}

Theorem~\ref{thm:EF1} implies the following corollary.

\begin{corollary}
In any repeated matching instance with up to four agents/goods, an \EFone\ repeated matching always exists.
\end{corollary}

\section{Are Fairness and Efficiency Compatible?} \label{sec:swef1}
In this section, we show that achieving the concepts of efficiency and fairness simultaneously is computationally intractable. In particular, we show in Theorem~\ref{thm:inapprox} below that even approximating the maximum social welfare of \EFone\ repeated matching is hard. Our proof is inspired by a reduction by~\citet{barman2019fair} but is more involved. Interestingly, it uses instances with constant valuations and comes in sharp contrast to achieving the two concepts separately. We remind the reader that, for such instances, an \EFone\ repeated matching can be computed in polynomial time by the techniques of \citet{biswas2018fair} while a polynomial time algorithm for computing social welfare maximizing repeated matchings follows by Theorem~\ref{thm:sw-non-dec}. 

\begin{theorem}\label{thm:inapprox}
For every constant $\epsilon>0$, approximating the maximum social welfare of \EFone\ repeated matchings on instances with $n$ agents/goods and $T$ rounds within a factor of $O\left(\min\{n^{1/3-\epsilon},T^{1-\epsilon}\}\right)$ is NP-hard.
\end{theorem}

\begin{proof}
We present a polynomial-time reduction, which, given a graph $G=(V,E)$, constructs a repeated matching instance $I(G)$ in which the maximum social welfare over \EFone\ repeated matchings is in $[K,K+1)$ if and only if the maximum independent set in graph $G$ has size $K$. Our construction leads to instances with $n\leq |V|^3$ agents/items and $T=|V|$ rounds. Then, the theorem follows by the next well-known result by \citet{Z07}.

\begin{theorem}[\citealp{Z07}]
For every constant $\epsilon>0$, approximating the maximum independent set of a graph $G=(V,E)$ within a factor of $|V|^{1-
\epsilon}$ is NP-hard.
\end{theorem}

Let $\delta$ be such that $0<\delta<|V|^{-2}$. Let $G=(V,E)$ be a graph. Without loss of generality, we can assume that $G$ has no isolated nodes, as the existence of such nodes just makes the independent set problem easier. Given graph $G=(V,E)$, the instance $I(G)$ has $T=|V|$ rounds and $n=(2|V|+1)|E|+1$ agents/items. For every edge $e\in E$, $I(G)$ has $2|V|+1$ {\em edge agents} identified as $(e,i)$ for $i=1, 2, ..., 2|V|+1$. There is also a {\em special agent} $s$. For every node $u\in V$, there is a {\em node item} $g_u$. The instance also has $n-|V|$ {\em dummy items}. For edge $e=(x,y)\in E$, $i \in [2|V|+1]$, and $t\in [T]$, the valuation of the edge agent $(e,i)$ for the $t^{\text{th}}$ copy of the node item $g_u$ is $v_{e,i}(g_u,t)=\delta$ if $u=x$ or $u=y$, and $v_{e,i}(g_u,t)=0$ otherwise. For node $u\in V$ and $t\in [T]$, the valuation of the special agent for the $t^{\text{th}}$ copy of the node item $g_u$ is $v_{s}(g_u,t)=1$. All agents have zero valuations for the dummy items. 

Let $K$ be the size of the maximum independent set in $G$. We claim that any \EFone\ repeated matching of $I(G)$ has social welfare less than $K+1$. This will follow by two observations for any \EFone\ repeated matching $A$. First, for every edge $e$, there is some $i\in [2|V|+1]$ such that the edge agent $(e,i)$ has value $0$. Assume that this is not true for edge $e=(x,y)$. Hence, $2|V|+1$ copies of the node items $g_x$ and $g_y$ have been given to the edge agents corresponding to edge $e$. However, we only have $|V|$ copies of each item. Second, consider the node items the special agent gets. As for each edge $e=(x,y)$, there is some agent $(e,i)$ who has zero value, the special agent can get at most one copy of node items $g_x$ or $g_y$. 
As this holds for every $e\in E$, the node items that the special agent gets correspond to the nodes in an independent set in $G$. Hence, her value is at most $K$. The total value the edge agents get from the $|V|$ node items they get is at most $|V|^2\cdot \delta<1$. Hence, the social welfare of repeated matching $A$ is less than $K+1$.

We now show that an \EFone\ repeated matching of social welfare in $[K,K+1)$ does exist, when the graph $G$ has an independent set $S$ of size $K$. First, the special agent gets a single copy of node item $g_x$ for each $x\in S$. The remaining copies of the node items are given to the edge agents in such a way that each edge agent corresponding to edge $e=(x,y)$ gets at most one copy of either $g_x$ or $g_y$. This is always possible, since for every edge $e=(x,y)$, there are $2|V|+1$ edge agents to get at most one copy of either node item $g_x$ or node item $g_y$. Then, the copies of the dummy items are distributed so that each agent has exactly $|V|$ item copies. As every edge agent has at most one copy of a node item, the \EFone\ conditions between any two of them are satisfied. Finally, the \EFone\  is satisfied between any edge agent and the special agent since the special agent gets at most one item copy for which the edge agent has positive value.
\end{proof}

\section{Swap Envy-Freeness}\label{sec:swapEF}
We now specifically turn our attention to repeated matching instances with mixed items. Consider the following instance with $n=2$ and $T=1$. One of the items is a good and the other is a chore. There are exactly two possible matchings. In either, the classical extension of \EFone\ for mixed items from the fair division literature (e.g., see \citealp{ACIW22}), which requires that the value of an agent is higher than that of another either by removing a single item from either one of the two bundles, is not satisfied.
Motivated by this simple example, we propose and investigate an alternate notion of fairness to \EFone\ for repeated matchings, which we call {\em swap envy-freeness} (\swapEF). 

\begin{definition}[\swapEF]
Let $I=\langle \A,\G,\{v_i\}_{i\in \A},T\rangle$ be a repeated matching instance with mixed items. A repeated matching $A=(A_1, ..., A_n)$ in $I$ is \swapEF\ if for every pair of agents $i,j \in \A$, either (i) or (ii) is true: 
\begin{enumerate}[label=(\roman*)]
    \item $v_i(A_i)\geq v_i(A_j)$;
    \item There exist items $g_i\in A_i$ and $g_j\in A_j$ such that $v_i(A_i\cup \{g_j\}\setminus \{g_i\}) \geq v_i(A_j\cup \{g_i\}\setminus \{g_j\})$. 
\end{enumerate}
\end{definition}
Condition (ii) requires that the value agent $i$ has for her bundle $A_i$ after replacing a copy of item $g_i$ with an extra copy of item $g_j$ is at least as high as her value for the bundle $A_j$ of agent $j$ after exchanging a copy of item $g_j$ with a copy of item $g_i$. We first find that Algorithm \ref{alg:identical} successfully finds a \swapEF\ repeated matching, even without the non-negativity constraint on valuations (the rank definition can be trivially adapted).

\begin{lemma}\label{lem:idswap}
Given a repeated matching instance $I=\langle \A,\G,v,T\rangle$ with identical valuations, the repeated matching returned by Algorithm \ref{alg:identical} is \swapEF.
\end{lemma}

\begin{proof}
First observe that if $T$ is an integer multiple of $n$, the repeated matching computed by Algorithm~\ref{alg:identical} creates no envy to any agent and, hence, it is \swapEF\ as well. Now, assume that $T$ is not an integer multiple of $n$; the algorithm will execute $T\bmod n$ round-robin phases in this case. Denote by $g_{i,t}$ the item agent $i$ gets in the round-robin phase $t\in \{1,2, ..., T\bmod n\}$. Agent $i$ gets $\lceil T/n\rceil$ copies of each of these items, while it uses only $\lfloor T/n\rfloor$ copies of the rest. Then, for any pair of agents $i$ and $j$, observe that
\begin{align*}
    v(A_i)-v(A_j) &= \sum_{t=1}^{T\bmod n}{(v(g_{i,t},\lceil T/n\rceil)-v(g_{j,t},\lceil T/n\rceil))}.
\end{align*}
If $i<j$, then it is also $v(g_{i,t},\lceil T/n\rceil)\geq v(g_{j,t},\lceil T/n\rceil)$ for every round-robin phase $t$, which implies that $v(A_i)\geq v(A_j)$. The inequality is clear if both agents $i$ and $j$ get a copy of the same item in phase $t$. If this is not the case, the item agent $i$ picks has lower rank than the item agent $j$ picks later. This implies that $v(g_{i,t},\lceil T/n\rceil)\geq v(g_{j,t},\lceil T/n\rceil)$, too. 

Now assume that $i>j$. By the argument above, we also get $v(A_i)\geq v(A_j)$ when agents $i$ and $j$ get a copy of the same item in each round. So, in the following, let us assume that this is not the case and denote by $t_1$ and $t_2$ the first and the last round-robin phase in which agents $i$ and $j$ get different items. Then, 
\begin{align*}
    &v(A_i\cup \{g_{j,t_1}\} \setminus \{g_{i,t_2}\})-v(A_j\cup \{g_{i,t_2}\} \setminus \{ g_{j,t_1}\})\\ 
    &\quad= \sum_{t=1}^{t_1-1}{(v(g_{i,t},\lceil T/n\rceil)-v(g_{j,t},\lceil T/n\rceil))}\\
    &\quad+v(g_{j,t_1},\lceil T/n\rceil)+\sum_{t=t_1}^{t_2-1}{(v(g_{i,t},\lceil T/n\rceil)-v(g_{j,t+1},\lceil T/n\rceil))}-v(g_{i,t_2},\lceil T/n\rceil)\\
    &\quad+\sum_{t=t_2+1}^{T\bmod n}{(v(g_{i,t},\lceil T/n\rceil)-v(g_{j,t},\lceil T/n\rceil))}\\
    &\quad\geq v(g_{j,t_1},\lceil T/n\rceil)-v(g_{i,t_2},\lceil T/n\rceil)\geq 0.
\end{align*}
The first inequality follows since the first and third sums are equal to $0$ and the second one is non-negative. This is due to the following observations. First, notice that, by definition, both agents $i$ and $j$ get a copy of the same item in phases from $1$ to $t_1-1$ and from $t_2+1$ to $T\bmod n$. Second, notice that the item $g_{i,t}$ that agent $i$ picks in round-robin phase $t$ is either the same with the one that agent $j$ picks in the next round-robin phase $t+1$ or one that has lower rank (and, thus, is at least as preferable). The second inequality is due to the fact that the item that agent $j$ picks in the round-robin phase $t_1$ is at least as preferable to the one agent $i$ picks later in the round-robin phase $t_2\geq t_1$.

We have established the \swapEF\ requirements in any case and the proof is complete.
\end{proof}

We now turn our attention to general valuations. 


\begin{theorem}\label{thm:swapEF}
Given a repeated matching instance $I$ with mixed items, $n$ agents, and $T$ rounds such that $T\bmod n \in\{0,1,2,n-2,n-1\}$, a \swapEF\ repeated matching exists and can be computed in polynomial time.
\end{theorem}


The proof of Theorem~\ref{thm:swapEF} uses Algorithm~\ref{alg:ef1-1} from Section~\ref{sec:fairness} for instances with $T\bmod n\in \{0,1,2\}$. For instances with $T\bmod n\in \{n-2,n-1\}$, we use an extension of Algorithm~\ref{alg:ef1-2} from Section~\ref{sec:fairness}, which runs an additional reverse round robin phase to remove one more distinct item from each agent when $T\bmod n = n-2$. We refer to this as Algorithm~\ref{alg:swapEF}; the lines~\ref{alg3:L9}-\ref{alg3:L15} implement the reverse round-robin phase, while the lines~\ref{alg3:L2}-\ref{alg3:L7} are identical to Algorithm~\ref{alg:ef1-2}. 

\begin{algorithm}[!ht]
    \KwIn{Instance $I=\langle \A,\G,\{v_i\}_{i\in \A},T\rangle$ with $|\A|=n$ and $T\bmod n \in \{n-1, n-2\}$}
    \KwOut{A repeated matching $A$}
    $f(i,g)\gets \lceil T/n\rceil, \forall i\in \A,\forall g\in G$\;\label{alg3:L2}
    $P\gets \G$\;\label{alg3:L3}
    \For{$i=1$ to $n$}{\label{alg3:L4}
        $\widehat{g}\gets \argmin_{g\in P} v_i(g, \lceil T/n\rceil)$\;\label{alg3:L5}
        $f(i,\widehat{g})\gets f(i,\widehat{g})-1$\;\label{alg3:L6}
        $P\gets P \setminus \{\widehat{g}\}$\;\label{alg3:L7}
    }
    
     \If{$T\bmod n = n-2$}{\label{alg3:L9}
        $P\gets \G$\;
        \For{$i=n$ to $1$}{\label{alg3:L11}
            $\widehat{g}\gets \argmin_{g\in P} v_i(g, f(i,g))$\;\label{alg3:L13}
        $f(i,\widehat{g})\gets f(i,\widehat{g})-1$\;\label{alg3:L14}
        $P\gets P \setminus \{\widehat{g}\}$\;\label{alg3:L15}
        }
    }     
    $A\gets \text{GenerateFromFreq}(f)$\;

\caption{Computing a \swapEF\ repeated matching} \label{alg:swapEF}
\end{algorithm}

The properties of Algorithms~\ref{alg:ef1-1} and~\ref{alg:swapEF} regarding \swapEF\ are given by the next two lemmas, which, together with the fact that both algorithms run in polynomial time, complete the proof of Theorem~\ref{thm:swapEF}.

\begin{lemma}\label{lem:alg1-swapEF}
The repeated matching $A=(A_1, ..., A_n)$ produced by Algorithm~\ref{alg:ef1-1} is \swapEF.
\end{lemma}

\begin{proof}
Let $S$ denote the multiset that contains each item with multiplicity $\lfloor T/n \rfloor$. If $T\bmod n=0$, then $A_i=S$ for every agent $i$ and, hence, the agents are not envious of each other. If $T\bmod n=1$, the final repeated matching is obtained after the execution of the round-robin phase. Consider two agents $i$ and $j$ and let $g_i$ and $g_j$ be the items the two agents get in this phase, respectively. If $v_i(g_i,\lceil T/n\rceil)\geq v_i(g_j,\lceil T/n\rceil)$, then \begin{align*}
v_i(A_i)&=v_i(S)+v_i(g_i,\lceil T/n \rceil) \geq v_i(S)+v_i(g_j,\lceil T/n \rceil)=v_i(A_j).
\end{align*}
Otherwise, if $v_i(g_i,\lceil T/n\rceil)< v_i(g_j,\lceil T/n\rceil)$, then 
\begin{align*}
v_i(A_i\cup\{g_j\}\setminus \{g_i\}) &=v_i(S)+v_i(g_j,\lceil T/n \rceil) >v_i(S)+v_i(g_i,\lceil T/n \rceil)=v_i(A_j\cup\{g_i\}\setminus \{g_j\}).
\end{align*}
In both cases, the \swapEF\ conditions are satisfied.

If $T\bmod n=2$, the final repeated matching is obtained after the execution of the reverse round-robin phase. Consider two agents $i$ and $j$. Let $g_i^1$ and $g_j^1$ be the items agents $i$ and $j$ get in the round-robin phase and $g_i^2$ and $g_j^2$ be the items they get in the reverse round-robin phase, respectively. We distinguish between three cases. If $|\{g_i^1,g_i^2\}\cap \{g_j^1,g_j^2\}|=2$, then $A_i$ and $A_j$ are effectively identical and agent $i$ does not envy agent $j$. If $|\{g_i^1,g_i^2\}\cap \{g_j^1,g_j^2\}|=1$, assume, without loss of generality, that $g_i^1=g_j^2=g$ and observe that $A_i$ has $\lceil T/n \rceil$ copies of $g_i^2$ and $\lfloor T/n \rfloor$ copies of $g_j^1$ and $A_j$ has $\lceil T/n \rceil$ copies of $g_j^1$ and $\lfloor T/n \rfloor$ copies of $g_i^2$. If $v_i(g_i^2,\lceil T/n \rceil) \geq v_i(g_j^1,\lceil T/n \rceil)$, then 
\begin{align*}
v_i(A_i)&=v_i(S\cup \{g\})+v_i(g_i^2,\lceil T/n\rceil) \geq v_i(S\cup \{g\})+v_i(g_j^1,\lceil T/n\rceil)=v_i(A_j).
\end{align*}
Otherwise, if $v_i(g_i^2,\lceil T/n \rceil) < v_i(g_j^1,\lceil T/n \rceil)$, then 
\begin{align*}
v_i(A_i\cup \{g_j^1\}\setminus \{g_i^2\}) &= v_i(S)+v_i(g_j^1,\lceil T/n\rceil) >v_i(S)+v_i(g_i^2,\lceil T/n \rceil)=v_i(A_j\cup \{g_i^2\}\setminus \{g_j^1\}).
\end{align*}
So, the \swapEF\ conditions are satisfied.

It remains to consider the case where $g_i^1$, $g_i^2$, $g_j^1$, and $g_j^2$ are distinct. Then, $A_i$ contains $\lceil T/n\rceil$ copies of $g_i^1$ and $g_i^2$ and $\lfloor T/n \rfloor$ copies of $g_j^1$ and $g_j^2$ and $A_j$ contains $\lceil T/n\rceil$ copies of $g_j^1$ and $g_j^2$ and $\lfloor T/n \rfloor$ copies of $g_i^1$ and $g_i^2$. If $i<j$, agent $i$ acts before agent $j$ in the round-robin phase and, hence, $v_i(g_i^1,\lceil T/n\rceil)\geq v_i(g_j^1,\lceil T/n\rceil)$. If $v_i(g_i^2,\lceil T/n\rceil)\geq v_i(g_j^2,\lceil T/n\rceil)$, then \begin{align*}v_i(A_i)&=v_i(S)+v_i(g_i^1,\lceil T/n\rceil)+v_i(g_i^2,\lceil T/n\rceil) \geq v_i(S)+v_i(g_j^1,\lceil T/n\rceil)+v_i(g_j^2,\lceil T/n\rceil)=v_i(A_j),
\end{align*}
and agent $i$ does not envy agent $j$. Otherwise, if $v_i(g_i^2,\lceil T/n\rceil)< v_i(g_j^2,\lceil T/n\rceil)$, then 
\begin{align*}
v_i(A_i\cup \{g_j^2\}\setminus \{g_i^2\}) &=v_i(S)+v_i(g_i^1,\lceil T/n\rceil)+v_i(g_j^2,\lceil T/n\rceil)\\
&>v_i(S)+v_i(g_j^1,\lceil T/n\rceil)+v_i(g_i^2,\lceil T/n\rceil)=v_i(A_i\cup \{g_i^2\}\setminus \{g_j^2\}),
\end{align*}
and the \swapEF\ condition is satisfied. If $i>j$, agent $i$ acts before $j$ in the reverse round-robin phase and, hence, $v_i(g_i^2,\lceil T/n\rceil)\geq v_i(g_j^2,\lceil T/n\rceil)$. If $v_i(g_i^1,\lceil T/n\rceil)\geq v_i(g_j^1,\lceil T/n\rceil)$, then \begin{align*}
v_i(A_i)&=v_i(S)+v_i(g_i^1,\lceil T/n\rceil)+v_i(g_i^2,\lceil T/n\rceil) \geq v_i(S)+v_i(g_j^1,\lceil T/n\rceil)+v_i(g_j^2,\lceil T/n\rceil)=v_i(A_j),
\end{align*}
and agent $i$ does not envy agent $j$. Otherwise, if $v_i(g_i^1,\lceil T/n\rceil)< v_i(g_j^1,\lceil T/n\rceil)$, then 
\begin{align*}
v_i(A_i\cup \{g_j^1\}\setminus \{g_i^1\})
&=v_i(S)+v_i(g_i^2,\lceil T/n\rceil)+v_i(g_j^1,\lceil T/n\rceil)\\
&>v_i(S)+v_i(g_j^2,\lceil T/n\rceil)+v_i(g_i^1,\lceil T/n\rceil) =v_i(A_i\cup \{g_i^1\}\setminus \{g_j^1\}),
\end{align*}
and the \swapEF\ condition is again satisfied.
\end{proof}

\begin{lemma}\label{lem:alg2-swapEF}
The repeated matching $A=(A_1, ..., A_n)$ produced by Algorithm~\ref{alg:swapEF} is \swapEF.
\end{lemma}

\begin{proof}
Let $S$ denote the multiset that contains each item with multiplicity $\lceil T/n \rceil$. We first consider the case where $T\bmod n = n-1$. Let $i$ and $j$ be two agents and denote by $g_i$ and $g_j$ the items that are removed from their bundles in the round-robin phase. If $v_i(g_i,\lceil T/n\rceil)\geq v_i(g_j,\lceil T/n\rceil)$, then 
\begin{align*}
v_i(A_i\cup \{g_i\}\setminus \{g_j\}) &= v_i(S)-v_i(g_j, \lceil T/n \rceil) \geq v_i(S)-v_i(g_i,\lceil T/n \rceil)=v_i(A_j\cup \{g_j\}\setminus \{g_i\}).
\end{align*}
Otherwise, if $v_i(g_i,\lceil T/n\rceil)< v_i(g_j,\lceil T/n\rceil)$, then 
\begin{align*}
v_i(A_i)&=v_i(S)-v_i(g_i,\lceil T/n \rceil) >v_i(S)-v_i(g_j,\lceil T/n \rceil)=v_i(A_j).
\end{align*}

If $T\bmod n=n-2$, the final repeated matching is obtained after the execution of the reverse round-robin phase. Consider two agents $i$ and $j$. Let $g_i^1$ and $g_j^1$ be the items agents $i$ and $j$ remove in the round-robin phase and $g_i^2$ and $g_j^2$ be the items they remove in the reverse round-robin phase, respectively. We distinguish between three cases. If $|\{g_i^1,g_i^2\}\cap \{g_j^1,g_j^2\}|=2$, then $A_i$ and $A_j$ are identical and agent $i$ does not envy agent $j$. If $|\{g_i^1,g_i^2\}\cap \{g_j^1,g_j^2\}|=1$, assume, without loss of generality, that $g_i^1=g_j^2=g$ and observe that $A_i$ has $\lfloor T/n \rfloor$ copies of $g_i^2$ and $\lceil T/n \rceil$ copies of $g_j^1$ and $A_j$ has $\lfloor T/n \rfloor$ copies of $g_j^1$ and $\lceil T/n \rceil$ copies of $g_i^2$. If $v_i(g_i^2,\lceil T/n \rceil) \leq v_i(g_j^1,\lceil T/n \rceil)$, then 
\begin{align*}
v_i(A_i)&=v_i(S\setminus \{g\})-v_i(g_i^2,\lceil T/n\rceil) \geq v_i(S\setminus \{g\})-v_i(g_j^1,\lceil T/n\rceil)=v_i(A_j).
\end{align*}
Otherwise, if $v_i(g_i^2,\lceil T/n \rceil) > v_i(g_j^1,\lceil T/n \rceil)$, then 
\begin{align*}
v_i(A_i\cup \{g_i^2\}\setminus \{g_j^1\}) &= v_i(S)-v_i(g,\lceil T/n \rceil )-v_i(g_j^1,\lceil T/n\rceil)\\
&>v_i(S) -v_i(g,\lceil T/n \rceil) -v_i(g_i^2,\lceil T/n \rceil)
= v_i(A_j\cup \{g_j^1\} \setminus \{g_i^2\}).
\end{align*}
So, the \swapEF\ conditions are satisfied in this case.

It remains to consider the case where $g_i^1$, $g_i^2$, $g_j^1$, and $g_j^2$ are distinct. Then, $A_i$ contains $\lfloor T/n\rfloor$ copies of $g_i^1$ and $g_i^2$ and $\lceil T/n \rceil$ copies of $g_j^1$ and $g_j^2$ and $A_j$ contains $\lfloor T/n\rfloor$ copies of $g_j^1$ and $g_j^2$ and $\lceil T/n \rceil$ copies of $g_i^1$ and $g_i^2$. If $i<j$, agent $i$ acts before agent $j$ in the round-robin phase and, hence, $v_i(g_i^1,\lceil T/n\rceil)\leq v_i(g_j^1,\lceil T/n\rceil)$. If $v_i(g_i^2,\lceil T/n\rceil)\leq v_i(g_j^2,\lceil T/n\rceil)$, then \begin{align*}v_i(A_i)&=v_i(S)-v_i(g_i^1,\lceil T/n\rceil)-v_i(g_i^2,\lceil T/n\rceil)
\geq v_i(S)-v_i(g_j^1,\lceil T/n\rceil)-v_i(g_j^2,\lceil T/n\rceil)=v_i(A_j),
\end{align*}
and agent $i$ does not envy agent $j$. Otherwise, if $v_i(g_i^2,\lceil T/n\rceil)> v_i(g_j^2,\lceil T/n\rceil)$, then 
\begin{align*}
v_i(A_i\cup \{g_i^2\}\setminus \{g_j^2\}) &=v_i(S)-v_i(g_i^1,\lceil T/n\rceil)-v_i(g_j^2,\lceil T/n\rceil)\\
&>v_i(S)-v_i(g_j^1,\lceil T/n\rceil)-v_i(g_i^2,\lceil T/n\rceil) =v_i(A_i\cup \{g_j^2\}\setminus \{g_i^2\}),
\end{align*}
and the \swapEF\ condition is satisfied. If $i>j$, agent $i$ acts before $j$ in the reverse round-robin phase and, hence, $v_i(g_i^2,\lceil T/n\rceil)\leq v_i(g_j^2,\lceil T/n\rceil)$. If $v_i(g_i^1,\lceil T/n\rceil)\leq v_i(g_j^1,\lceil T/n\rceil)$, then \begin{align*}
v_i(A_i)&=v_i(S)-v_i(g_i^1,\lceil T/n\rceil)-v_i(g_i^2,\lceil T/n\rceil) \geq v_i(S)-v_i(g_j^1,\lceil T/n\rceil)-v_i(g_j^2,\lceil T/n\rceil)=v_i(A_j),
\end{align*}
and agent $i$ does not envy agent $j$. Otherwise, if $v_i(g_i^1,\lceil T/n\rceil)> v_i(g_j^1,\lceil T/n\rceil)$, then 
\begin{align*}
v_i(A_i\cup \{g_i^1\}\setminus \{g_j^1\})
&=v_i(S)-v_i(g_i^2,\lceil T/n\rceil)-v_i(g_j^1,\lceil T/n\rceil)\\
&>v_i(S)-v_i(g_j^2,\lceil T/n\rceil)-v_i(g_i^1,\lceil T/n\rceil)=v_i(A_i\cup \{g_j^1\}\setminus \{g_i^1\}),
\end{align*}
and the \swapEF\ condition is again satisfied.
\end{proof}

Theorem~\ref{thm:swapEF} implies the following corollary.

\begin{corollary}
In any repeated matching instance with mixed items and up to five agents/items, a \swapEF\ repeated matching always exists.
\end{corollary}

We conclude this section with a comparison of \EFone\ and \swapEF. While the two fairness notions have similar definitions, they are actually incomparable. Clearly, \swapEF\ does not imply \EFone\ as it is trivially satisfied in the simple motivating example with one good and one chore presented at the beginning of this section. However, given that we use largely the same algorithms for \swapEF\ as we did for \EFone, one may believe intuitively that for goods alone, \EFone\ implies \swapEF. This is not the case though. Consider an instance with three rounds and two agents with identical constant valuations $v(1,t)=3$ and $v(2,t)=2$ for two items. Giving item $1$ to one agent and item $2$ to the other for all three rounds is \EFone\ but not \swapEF.

\section{Open Problems}\label{sec:open}
Our work leaves several interesting open problems that deserve investigation. Understanding social welfare maximization is the first one. Is the problem hard for instances with two rounds? Recall that our hardness reduction in the proof of Theorem~\ref{thm:sw-hard} uses three rounds while the problem is in P for a single round. What about approximation algorithms when the items are goods and valuations are not necessarily monotone? Is a constant approximation ratio possible? Regarding fairness, the most important open question is whether \EFone\ repeated matchings exist for any instance with goods. Furthermore, is \EFone\ compatible with different notions of efficiency than the utilitarian social welfare we have used here? For example, what about the egalitarian or Nash social welfare? Is \EFone\  compatible with Pareto-efficiency? For instances with mixed items, do \swapEF\ repeated matchings always exist? Again, how do they interplay with Pareto-efficiency? In general, \swapEF\ deserves investigation in other fair division settings that involve mixed items.

\bibliographystyle{named.bst} 
\bibliography{main.bib}

\newpage
\appendix

\section{Social Welfare Maximization under Monotone Non-Increasing Valuations} \label{sec:swnoninc}

We now give an alternative proof to the next theorem. This is not a new result; we remind the reader that the proof already follows by well-known results on weighted bipartite $b$-matchings (see the discussion in Section~\ref{subsec:welfare-non-increasing}).

\begin{theorem}
Given a repeated instance with monotone non-increasing valuations, a repeated matching of maximum social welfare can be computed in polynomial time.
\end{theorem}

\begin{proof}
Consider a repeated matching instance $I=\langle \A, \G, \{v_i\}_{i\in \A},T\rangle$. We express the problem of computing a repeated matching of maximum social welfare as the following integer linear program:
\begin{align*}
\text{max }&\sum_{i\in \A}\sum_{g\in\G}\sum_{t=1}^T{x_{i,g,t}\cdot v_i(g,t)}\\
\text{ s.t.: } &\sum_{g\in \G, t\in [T]}x_{i,g,t}= T, &&\forall i\in \A \\
    &\sum_{i\in \A, t\in [T]}x_{i,g,t}= T, &&\forall g\in \G \\
    &x_{i,g,t}\geq x_{i,g,t+1}, &&\forall i\in\A, g\in G, t\in [T-1]\\
    &x_{i,g,t} \in \{0,1\}, &&\forall i\in\A,\,g\in \G,\,t\in [T] 
\end{align*}
The binary indicator variable $x_{i,g,t}$ denotes whether agent $i$ gets the $t^\text{th}$ copy of item $g$ ($x_{i,g,t}=1$) or not ($x_{i,g,t}=0$). Then, the objective is clearly to maximize the social welfare, the total value the agents get from their copies of items. The first set of constraints requires that each agent gets exactly $T$ copies of the items in the $T$ rounds. The second one requires that each item is assigned in all the $T$ rounds. The third one ensures that an agent can get her $(t+1)^{\text{th}}$ copy of an item only after she gets the $t^{\text{th}}$ copy. 

We now relax the integrality constraint by replacing $x_{i,g,t}\in \{0,1\}$ with $x_{i,g,t}\in [0,1]$. In this way, we get a linear program (LP). Well-known solvers, implementing variants of the {\em ellipsoid method}~\citep{GLS88,S86}, can solve this LP in polynomial time and compute an {\em extreme} solution. Consider such an extreme solution $x$ and, for the sake of contradiction, assume that it is non-integral. 

We first show that for every agent $i\in \A$ and item $g\in\G$, at most one variable $x_{i,g,t}$ can be non-integral. Assume otherwise for agent $i\in\A$ and item $g\in\G$, and let $t_1$ and $t_2$ be the maximum and minimum elements in set $\{t:0<x_{i,g,t}<1\}$. Let $\epsilon=\min\{1-x_{i,g,t_1}, x_{i,g,t_2}\}$ and consider the modified solution $x'$ with $x'_{i,g,t_1}=x_{i,g,t_1}+\epsilon$ and $x'_{i,g,t_2}=x_{i,g,t_2}-\epsilon$, while $x'$ has the same value with $x$ on any triplet different than $(i,g,t_1)$ and $(i,g,t_2)$. Due to the feasibility of $x$, the new solution $x'$ is clearly feasible. Furthermore, the objective value of $x'$ is (at least) as high as that of $x$ as it increases by $\epsilon\cdot v_i(g,t_1)$ and decreases by $\epsilon\cdot v_i(g,t_2)\leq \epsilon\cdot v_i(g,t_1)$. The last inequality follows since the valuations are monotone non-increasing. Hence, the solution $x'$ has optimal objective value as well, and, furthermore, at least one additional integral variable compared to $x$: $x'_{i,g,t_1}=1$ if $1-x_{i,g,t_1}\leq x_{i,g,t_2}$ and $x'_{i,g,t_2}=0$ otherwise. Thus, solution $x$ is not extreme, a contradiction.

Now, consider the bipartite graph $G=(\A,\G,E_x)$, where $E_x$ contains the edge $(i,g)$ if there exists $t$ such that $x_{i,g,t}$ has non-integer value. Observe that $G$ contains cycles. Indeed, if $G$ was a tree, some node $u$ in $\A\cup\G$ would have degree $1$. If $u\in \A$, then $\sum_{g\in \G,t\in [T]}{x_{u,g,t}}$ would include a single non-integer term (i.e., the weight $x_{u,g,t}$ of the single edge which is incident to node $u$). As $T$ is integer, it would then be $\sum_{g\in \G,t\in [T]}{x_{u,g,t}}\not=T$, violating the first LP constraint for agent $u$. If $u\in \G$, then $\sum_{i\in \A,t\in [T]}{x_{i,u,t}}$ would include a single non-integer term. Again, this would imply that $\sum_{i\in \A,t\in [T]}{x_{i,u,t}}\not=T$, violating the second LP constraint for item $u$.

Let $C$ be a cycle in $G$. Since $G$ is bipartite, $C$ has even length and its edges can be partitioned into two matchings $M_1$ and $M_2$. For an edge $(i,g)$ of $E_x$, let $t(i,g)$ be such that $x_{i,g,t(i,g)}$ is non-integer. Also, for a set of edges $M$ of $E_x$, define $V(M)=\sum_{(i,g)\in M}{x_{i,g,t(i,g)}\cdot v_i(g,t(i,g))}$ and, without loss of generality, assume that $V(M_1)\geq V(M_2)$. Observe that $V(M_1)$ and $V(M_2)$ are simply the contribution to the objective value by the triplets $(i,g,t(i,g))$ corresponding to the edges $(i,g)$ of $M_1$ and $M_2$, respectively. Now let \[\epsilon = \min\left\{1-\max_{(i,g)\in M_1}{x_{i,g,t(i,g)}}, \min_{(i,g)\in M_2}{x_{i,g,t(i,g)}}\right\}\] 
and modify the solution $x$ to a new solution $x'$ as follows:
\begin{itemize}
    \item $x'$ has the same value with $x$ on any triplet that does not correspond to $(i,g,t(i,g))$ for an edge $(i,g)$ of $C$.
    \item $x'_{i,g,t(i,g)}=x_{i,g,t(i,g)}+\epsilon$ for every $(i,g)\in M_1$, and
    \item $x'_{i,g,t(i,g)}=x_{i,g,t(i,g)}-\epsilon$ for every $(i,g)\in M_2$.
\end{itemize}

Clearly, the contribution of a triplet that does not correspond to triplet $(i,g,t(i,g))$ for an edge $(i,g)$ of $C$ to the objective value is the same under both solutions $x$ and $x'$. The contribution from the triplets $(i,g,t(i,g))$ corresponding to edges $(i,g)$ of $M_1$ increases by $\epsilon\cdot V(M_1)$ in $x'$ compared to $x$, and the contribution from the triplets $(i,g,t(i,g))$ corresponding to edges $(i,g)$ of $M_2$ decreases by $\epsilon\cdot V(M_2)\leq \epsilon\cdot V(M_1)$. Hence, the objective value of solution $x'$ is also optimal. Furthermore, solution $x'$ has at least one additional integer variable compared to $x$: indeed, observe that $x'_{i_1,g_1,t(i,g)}=1$ for some edge $(i_1,g_1)$ of $M_1$ if $1-\max_{(i,g)\in M_1}{x_{i,g,t(i,g)}}\leq  \min_{(i,g)\in M_2}{x_{i,g,t(i,g)}}$ and $x'_{i_2,g_2,t(i_2,g_2)}=0$ for some edge $(i_2,g_2)$ of $M_2$, otherwise. Thus, solution $x$ is not extreme, again a contradiction.
%
\end{proof}




\end{document}